\documentclass{fundam}
\usepackage{url} 
\usepackage[ruled,lined]{algorithm2e}
\usepackage{graphicx}
\usepackage{tikz}
\usetikzlibrary{automata, positioning, arrows}

\usepackage{amsmath} 
\usepackage{amssymb} 
\usepackage{mathtools}
\usepackage{enumitem}

\newcommand{\bigdoublecap}{%
  \mathop{\vphantom{\bigcap}\vcenter{\hbox{\text{%
    \ooalign{$\displaystyle\bigcap$\cr
             \hidewidth
             \raisebox{-.1ex}{\resizebox{.7\width}{.875\height}{$\displaystyle\bigcap$}}%
             \hidewidth\cr
    }%
  }}}}%
}

\newtheorem{defi}{Definition}
\newtheorem{rema}{Remark}
\newtheorem{prop}{Proposition}
\newtheorem{teo}{Theorem}
\newtheorem{col}{Corollary}

\usepackage{graphicx}

\definecolor{aliceblue}{rgb}{0.94, 0.97, 1.0}
\definecolor{anti-flashwhite}{rgb}{0.90, 0.90, 0.91}

\begin{document}

%
\setcounter{page}{1}
\publyear{2021}
\papernumber{0001}
\volume{178}
\issue{1}
%

\title{On Typical Hesitant Fuzzy Languages and Automata}

\address{valdigleis.costa@univasf.edu.br}

\author{Valdigleis S. Costa\\
Universidade Federal do Vale do São Francisco \\
Colegiado de Ciência da Computação \\ Salgueiro-PE, Brazil\\
valdigleis.costa{@}univasf.edu.br
\and Benjamín C. Bedregal \thanks{Thanks to CNPQ for the research funding granted through the project 311429/2020-3.} 
\and Regivan H. N. Santiago \thanks{Thanks to CNPQ for the research funding granted through the project 312053/2018-5.}\\
Universidade Federal do Rio Grande do Norte \\
Departamento de Informática e Matemática Aplicada \\ Natal-RN, Brazil \\
bedregal{@}dimap.ufrn.br \\
regivan{@}dimap.ufrn.br } 

\maketitle

\runninghead{Valdigleis, Benjamín $\&$ Regivan}{On THFL and Automata}

\begin{abstract}
 The idea of nondeterministic typical hesitant fuzzy automata is a generalization of the fuzzy automata presented by Costa and Bedregal. This paper, presents the sufficient and necessary conditions for a typical hesitant fuzzy language to be computed by nondeterministic typical hesitant fuzzy automata. Besides, the paper introduces a new class of Typical Hesitant Fuzzy Automata with crisp transitions, and we will show that this new class is equivalent to the original class introduced by Costa and Bedregal.
\end{abstract}


\section{Introduction}

The fuzzy computation theory emerged as a model based on fuzzy sets \cite{zadeh1965} capable of extrapolating the Church's thesis \cite{silva2016}. The machine models more studied by fuzzy computation theory are fuzzy Turing machines \cite{silva2016, bedregal2008TM} and the fuzzy finite automata \cite{hopcroft2006}. Finite automata are a computational model that has finite memory limitation. This model is widely used for modeling applications in hardware and software \cite{Farias2017}, and it is also essential for building compilers \cite{hopcroft2006, aho2003}. In recent years, several generalizations of the concept of finite automata have been presented, such as fuzzy automata  \cite{Stanimirovic2018,  wei2018, mordeson2002, mizumoto1969, Costa2018}, probabilistic automata \cite{rabin1963, paz2014, abney1999}, and quantum automata  \cite{ying2000, moore2000, qiu2004, qiu2001, hirvensalo2011}. 

With the development of the various extensions for fuzzy sets, some generalizations for fuzzy automata have been presented, for example, interval-valued fuzzy automata \cite{ravi2010, ravi2010myhill} and intuitionistic fuzzy automata \cite{choubey2009}. Recently, Costa and Bedregal in \cite{valdigleis2019} using the concepts of typical hesitant fuzzy set \cite{torra01, benja2014}, introduced the concept of Nondeterministic Typical Hesitant Fuzzy Automata and presented a subclass, called Deterministic Typical Hesitant Fuzzy Automata, which generalizes the notion of deterministic finite automata. They showed how it is possible to obtain a Deterministic Typical Hesitant Fuzzy Automata from a Nondeterministic Typical Hesitant Fuzzy Automata. However, the removal of the nondeterminism presented by Costa and Bedregal does not preserve the language.

Costa and Bedregal defined nondeterministic typical hesitant fuzzy automata as machines capable to compute typical hesitant fuzzy language \cite{valdigleis2019}. However, it does not characterize a class for these languages. Thus an open question exists, \textit{every typical hesitant fuzzy language can be computed by a typical hesitant fuzzy automaton}?
 
In this article, the theory of typical hesitant fuzzy automata will be strengthened, characterizing the languages computed by nondeterministic typical hesitant fuzzy automata. Moreover, we will show the nondeterminism does not increase the power of typical hesitant fuzzy automata.  This work has the following division, first this introduction, then in Section 2, presents the mathematical basis for this work. Section 3, presents a characterization for the languages computed by typical hesitant fuzzy automata. Section 4, display a new class of typical hesitant fuzzy automata and some results.

\section{Preliminaries}

In this section, we present all the basic definitions and notations sed throughout the text.

\subsection{Languages and Finite Automata}

As said in \cite{hopcroft2006},  an alphabet is any finite non-empty set $\Sigma$. The elements from $\Sigma$ are called letters, and a word on $\Sigma$ is any finite sequence of letters. The symbol $\lambda$ denotes the empty word, i.e., the word without letters from $\Sigma$. The set $\Sigma^*$ is a free monoid generated by $\Sigma$ concerning the operation of concatenation \cite{valdigleis2019}. The set $\Sigma^+ = \Sigma^* - \{\lambda\}$, and any $L \subseteq \Sigma^*$ is called language.

\begin{defi}\label{defi:AFD}
    \cite{hopcroft2006} A deterministic finite automaton (DFA) is a quintuple $A = \langle Q, \Sigma, \delta, q_0, F \rangle$ where $Q$ is finite non-empty set of states, $\delta: Q \times \Sigma \rightarrow Q$ is a the transition function\footnote{This paper always assumes complete (N)DFA, meaning that the transition function is total.}, $q_0 \in Q$ is the initial state and $F \subseteq Q$ is the set of final states.
\end{defi}

The transition function $\delta$ can extend into a function $\widehat{\delta}: Q \times \Sigma^* \rightarrow Q$ by the following recursion:
\begin{eqnarray}\label{eq:ExtensaoDaFuncaoTransicaoDelta}
    \widehat{\delta}(q, \lambda)& = & q \\
    \widehat{\delta}(q, wa)& = & \delta(\widehat{\delta}(q, w), a)	
\end{eqnarray}
where $a \in \Sigma$ and $w \in \Sigma^*$. 

\begin{defi}
    The language computed by a DFA $A$ is the set $L(A) = \{w \in \Sigma^* \mid \widehat{\delta}(q, w) \in F\}$.
\end{defi}

As said in \cite{levelt2008}, if $w = a_1\cdots a_n \in \Sigma^+$ and $w \in L(A)$ for some DFA $A$, then there exists a finite sequence of states $q_1, \cdots, q_n$ such that $\delta(q_0, a_1) = q_1, \cdots, \delta(q_{n-1}, a_n) = q_n$ with $q_n \in F$. On the other hand, $\lambda \in L(A)$ if and only if $q_0 \in F$.

\begin{defi}
    \cite{hirvensalo2011} A language $L$ is regular, whenever $L$ is finite or $L$ is obtained from regular languages $L_1$ and $L_2$ by either finite union, concatenation, or Kleene closure\footnote{For more details about union, concatenation and Kleene closure see \cite{hopcroft2006}}.
\end{defi}

\begin{teo}\label{teo:regular}
    \cite{hirvensalo2011} A language $L$ is regular if and only if a DFA compute it.
\end{teo}

\begin{rema}\label{rema:NumeroDeLinguagens}
	Notice that by the Chomsky’s hierarchy \cite{Costa2018}, as mentioned in proof of Theorem 2 in \cite{rabin1963}, there is only an enumerable set of regular languages.
\end{rema}

Another well-known type of finite automata is the nondeterministic finite automaton, defined below.

\begin{defi}\label{defi:AFN}
    \cite{hopcroft2006} A nondeterministic finite automaton (NFA) is a quintuple $N = \langle Q, \Sigma, \delta_N, q_0, F \rangle$ where $Q, \Sigma, q_0$ and $F$ are equal to Definition \ref{defi:AFD}, and $\delta_N: Q \times \Sigma \rightarrow 2^Q$ is the nondeterministic transition function.
\end{defi}

As discussed by Hopcroft \textit{et al.} in \cite{hopcroft2006}, it is clear that every AFD is an AFN where the inequality $\#\delta(q, a) \leq 1$ is satisfied for every pair $(q, a) \in Q \times \Sigma$, where $\#$ denote the cardinality of sets. The function $\delta_N$ can be extended into a function $\widehat{\delta_N}: Q \times \Sigma^* \rightarrow 2^Q$ by the following recursion:
\begin{eqnarray}\label{eq:ExtensaoDaFuncaoTransicaoDeltaN}
    \widehat{\delta_N}(q, \lambda)& = & q \\
    \widehat{\delta_N}(q, wa)& = & \bigcup_{q \in \widehat{\delta_N}(q, w)} \{\delta_N(q, a)\}	
\end{eqnarray}
for all $a \in \Sigma$ and $w \in \Sigma^*$. 

\begin{defi}
    The language computed by NFA $A$ is the set $L(A) = \{w \in \Sigma^* \mid \widehat{\delta}(q, w) \cap F = \emptyset\}$.
\end{defi}

According to \cite{samuel1974}, for any $n \in \mathbb{N}$, there are $n$-state NFAs recognizing languages which cannot be recognized by any DFA with less than $2^n$ states.

\begin{teo}\label{teo:regular2}
    \cite{hopcroft2006} A language $L$ is regular if and only if a NFA compute it.
\end{teo}

\subsection{THFE, THFL, and THFA}

According to \cite{torra01, matzenauer2019}, a HFS is defined in terms of a function which return sets of membership degrees for each element of their domain $\mathcal{U} \neq \emptyset$. In 2014, Bedregal \textit{et al.} \cite{benja2014} introduced a particular case of HFS, called Typical Hesitant Fuzzy Set, or simply THFS, which considers some restrictions.

\begin{defi}\label{def:conjuntoshesitantestipicos}
	\cite[Definition 8]{benja2014} Let $\mathbb{H} \subseteq 2^{[0,1]}$ be the set of all finite non-empty subsets of the interval [0,1], and let $\mathcal{U}$ be a non-empty set. A THFS on $\mathcal{U}$ is  a function $\psi: \mathcal{U} \rightarrow \mathbb{H}$.
\end{defi}

\begin{rema}
    Here will be considered that $\langle [0,1], \vee, \wedge, 0, 1\rangle$ is a distributive lattice concerning the usual order $\leq$ on real numbers.
\end{rema}

Each $X \in \mathbb{H}$ is called a Typical Hesitant Fuzzy Element (THFE). The set $\mathbb{H}_1 = \{X \in \mathbb{H} \mid \#X = 1\}$ is called of degenerate elements set.  Several operators on $\mathbb{H}$ were proposed in \cite{benja2014, matzenauer2019, xia01, rod01, gar01, santos2015, bedregal2014typical}. In particular, Costa and Bedregal \cite{valdigleis2019}, have presented the inf-combination and sup-combination.

\begin{defi}\label{def:infcombinacao}
	\cite[Definition 3.1.]{valdigleis2019} For $X,Y \in \mathbb{H}$ the function $\otimes: \mathbb{H} \times \mathbb{H} \rightarrow \mathbb{H}$ is computed by,
	\begin{equation}
	    X \otimes Y = \{x \wedge y \mid x \in X, y \in Y\}
	\end{equation}
	is called inf-combination of $X$ and $Y$.
\end{defi}

\begin{defi}\label{def:supcombinacao}
	\cite[Definition 3.2.]{valdigleis2019} For $X,Y \in \mathbb{H}$ the function $\sqcup: \mathbb{H} \times \mathbb{H} \rightarrow \mathbb{H}$ is computed by,
	\begin{equation}
	    X \sqcup Y = \{x \vee y \mid x \in X, y \in Y\}
	\end{equation}
	is called sup-combination of $X$ and $Y$.
\end{defi}


The operations $\wedge$ and $\vee$ are, respectively, the operations of infimum and supremum on the distributive lattice $[0,1]$, for complete notions related to partial order and lattice theory, the reader can refer to \cite{gratzer2011lattice}. According to \cite{valdigleis2019}, $\mathbb{H}$ has the following properties:

\begin{enumerate}[leftmargin=*,label=\normalfont{($H$\arabic*)}, labelsep=4.9mm]
	\item The structures $\langle \mathbb{H}, \otimes, \{1\} \rangle$ and $\langle \mathbb{H}, \sqcup, \{0\} \rangle$ are commutative and idempotent monoids.
	\item The element $\{1\}$ is an annihilator of $\sqcup$, and $\{0\}$ is an annihilator of $\otimes$.
	\item $\otimes$ distribute over $\sqcup$ and $\sqcup$ distribute over $\otimes$.
\end{enumerate}

According to \cite{valdigleis2019} since $\langle \mathbb{H}, \sqcup, \{0\} \rangle$ is a monoid, the operation $\sqcup$ is extended for the $n$-dimensional case. 

\begin{defi}
    \cite{valdigleis2019} Given $X_1, X_2, \cdots, X_n \in \mathbb{H}$,
    \begin{eqnarray}\label{eq:SupCombinacaoNarria}
        \bigsqcup_{i = 1}^n X_i = \Big(\bigsqcup_{i=1}^{n-1} X_i\Big)\sqcup X_n
    \end{eqnarray}
\end{defi}

\begin{rema}\label{rema:elementosGerados}
    Since $\otimes$ and $\sqcup$ are both commutative, associative, and idempotent, given any finite non-empty set $\kappa \subset \mathbb{H}$, the set of elements generated by the (sup) inf-combination on the set $\kappa$ is also finite.
\end{rema}

The ordering problem THFS and THFE have yet been studied in \cite{gar01, santos2015, xu2011, benja2021}. Now we propose a new relation on $\mathbb{H}$ based in the sup-combination, and we will show that this order generalizes the usual order of real numbers.

\begin{defi}\label{Defi:Order}
    Given $X, Y \in \mathbb{H},X \sqsubseteq Y \Longleftrightarrow X \sqcup Y = Y$
\end{defi}

\begin{prop}\label{prop:Ordem}
    $\sqsubseteq$ is a partial order on $\mathbb{H}$.
\end{prop}

\begin{proof}
    Consider that $X, Y, Z \in \mathbb{H}$ so: 
    \begin{itemize}
        \item[(i)] we have that, $X \sqcup X \stackrel{(H1)}{=} X$. So $\sqsubseteq$ is reflexive;
        \item[(ii)] suppose that $X \sqsubseteq Y$ and $Y \sqsubseteq X$, then $X \stackrel{Hyp.}{=} Y \sqcup X  \stackrel{(H1)}{=} X \sqcup Y \stackrel{Hyp.}{=} Y$ so $\sqcup$ is anti-symmetric and
        \item[(iii)] suppose that $X \sqsubseteq Y$ and $Y \sqsubseteq Z$, thus we have that, $X \sqcup Z \stackrel{Hyp.}{=} X \sqcup (Y \sqcup Z) \stackrel{(H1)}{=} (X \sqcup Y) \sqcup Z \stackrel{Hyp.}{=} Y \sqcup Z \stackrel{Hyp.}{=} Z$
        therefore, $X \sqsubseteq Z$, so $\sqsubseteq$ is transitive.
    \end{itemize}
    Since that $\sqsubseteq$ is reflexive, anti-symmetric and transitive, the relation $\sqsubseteq$ is a partial order on $\mathbb{H}$.
\end{proof}

\begin{teo}
    The order $\sqsubseteq$ generalize the usual order $\leq$ on $[0,1]$.
\end{teo}

\begin{proof}
    Let $\mathbb{H}_1$ the set of degenerate elements of $\mathbb{H}$, for all $\{x\}, \{y\} \in \mathbb{H}_1$ we have that,
    $$\{x\} \sqsubseteq \{y\} \Longleftrightarrow  \{x\} \sqcup \{y\} = \{y\} \Longleftrightarrow x \vee y = y \Longleftrightarrow x \leq y$$
    and $x,y \in [0,1]$, thus completing the proof.
\end{proof}

Moreover, the following properties are easily verified.

\begin{enumerate}[leftmargin=*,label=\normalfont{($R$\arabic*)}, labelsep=4.9mm]
    \item If $X \sqsubseteq Y$, then $(X \sqcup Z) \sqsubseteq (Y \sqcup Z)$.
    \item If $X \sqsubseteq Y$, then $(X \otimes Z) \sqsubseteq (Y \otimes Z)$.
    \item $\{0\} \sqsubseteq X \sqsubseteq \{1\}$ for all $X \in \mathbb{H}$.
\end{enumerate}

\begin{prop}\label{prop:AbsorcaoDireita}
    For $\circ \in \{\otimes, \sqcup\}$, if $X \sqsubseteq Y$, then $X \sqsubseteq (X \circ Y)$.
\end{prop}

\begin{proof}
    Suppose that $X \sqsubseteq Y$ so by $(R1)$ and $(R2)$ we have that  $(X \circ X) \sqsubseteq (X \circ Y)$ but by $(H1)$,  $X \circ X = X$, therefore,  $X \sqsubseteq (X \circ Y)$.
\end{proof}

\begin{rema}\label{rema:MenorQueOperacao}
    Notice that for all $X, Y \in \mathbb{H}$, $X \sqcup (X \sqcup Y) = X \sqcup Y$. Therefore, $X \sqsubseteq (X \sqcup Y)$.
\end{rema}

\begin{teo}\label{teo:CotaSuperior}
	Let $\{X_i\}_{i \in I}$ be a finite family of elements of $\mathbb{H}$, then for any $X_j$, with $j \in I$,  $X_j \sqsubseteq \bigsqcup_{i \in I} X_i$.
\end{teo}

\begin{proof}
	By remark \ref{rema:MenorQueOperacao},  commutativity and associativity in $(H1)$.
\end{proof}

\begin{defi}\label{def:ImagemDominio}
	Let $\mathcal{U}$ be a nonempty set, $f: \mathcal{U} \rightarrow \mathbb{H}$  and $k \in \mathbb{H}$. Then we have the following sets:
	\begin{itemize}
		\item[(i)] $R_{f} = \{X \in \mathbb{H} \mid f(x) = X, x \in \mathcal{U}\}$. 
		\item[(ii)] $S^k_{f} = \{x \in \mathcal{U} \mid k  \sqsubseteq f(x)\}$
	\end{itemize}
\end{defi}

\begin{defi}
    A Typical Hesitant Fuzzy Language, or simply THFL, is any THFS $f: \Sigma^* \rightarrow \mathbb{H}$. The set of all Typical Hesitant Fuzzy Languages is denoted by $\mathbb{T}$.
\end{defi}

Recently, Costa and Bedregal in \cite{valdigleis2019} introduced a new generalization of fuzzy automata called Nondeterministic Typical Hesitant Fuzzy Automata.

\begin{defi}\label{defi:AFHT}
    \cite{valdigleis2019} A Nondeterministic Typical Hesitant Fuzzy Automaton, or simply NTHFA, is a quintuple $M = \langle Q, \Sigma, \psi, q_0, \mathcal{F} \rangle$ where $Q$ is a finite non-empty set of states, $\Sigma$ is an alphabet, $\psi: Q \times \Sigma \times Q \rightarrow \mathbb{H}$ is a THFS, $q_0$ is initial state and $\mathcal{F}: Q \rightarrow [0,1]$ is the THFS on $Q$ of final states.
\end{defi}

\begin{defi}
    \cite{valdigleis2019} For any NTHFA $M$ the functions $\psi$ is extended into a function $\widehat{\psi}: Q \times \Sigma^* \times Q \rightarrow \mathbb{H}$ using the recursion:
    \begin{equation}
        \widehat{\psi}(q, \lambda, q') = \left\{\begin{array}{rl}	\{0\}, & \hbox{if } q \neq q'\\	\{1\},  & \hbox{else}\end{array}\right.
    \end{equation}
    \begin{equation}\label{eq:part2}
        \widehat{\psi}(q, wa, q') = \bigsqcup_{q'' \in Q} \Big(\widehat{\psi}(q, w, q'') \otimes \psi(q'', a, q')\Big)
    \end{equation}
\end{defi}

\begin{defi}\label{defi:linguagemH}
    \cite{valdigleis2019} Let $M$ be a  NTHF, then $M$ computes the THFL $f_M : \Sigma^* \rightarrow \mathbb{H}$ is defined by,
    \begin{equation}
        f_M(w) = \bigsqcup_{q \in Q} \Big(\widehat{\psi}(q_0, w, q) \otimes \mathcal{F}(q) \Big)
    \end{equation}
\end{defi}

\begin{rema}\label{rema:ValorParaLambda}
	By definition \ref{defi:linguagemH} it's posible to deduce that, for all $M$, $f_M(\lambda) = \mathcal{F}(q_0)$.
\end{rema}

\section{A Characterization of the Languages Computed by NTHFA}

This paper will denote the set of all the THFL computed by NTHFA by $\mathbb{T}_\mathcal{R}$. Now consider the following definition.

\begin{defi}
    Let $\Sigma$ be an alphabet and let $f_1: \Sigma^* \rightarrow \mathbb{H}$ and $\mathcal{L}_2: \Sigma^* \rightarrow \mathbb{H}$ two be THFL such that $f_1, f_2 \in \mathbb{T}_\mathcal{R}$, the $\mathbb{H}$-union of $f_1$ and $f_2$, denoted by $f_1 \uplus f_2$, is given by,
	\begin{equation}\label{eq:Uniao}
		f_1 \uplus f_2(w) = f_1(w) \sqcup f_2(w)
	\end{equation}
	for all $w \in \Sigma^*$.
\end{defi}

\begin{teo}\label{teo:Uniao}
	If $f_1, f_2 \in \mathbb{T}_\mathcal{R}$, then $f_1 \uplus f_2 \in \mathbb{T}_\mathcal{R}$.
\end{teo}

\begin{proof}
    Assume that  $f:\Sigma^* \rightarrow \mathbb{H}$ and $f_2:\Sigma^* \rightarrow \mathbb{H}$ belongs to $\mathbb{T}_\mathcal{R}$, so there exists $M_1 = \langle S, \Sigma, \psi_1, s_0, \mathcal{F}_1 \rangle$ and $M_2 = \langle P, \Sigma, \psi_2, p_0, \mathcal{F}_2 \rangle$ such that $f_1 = f_{M_1}$ and $f_2 = f_{M_2}$. Moreover, without loss of generality, assume that $S \cap P = \emptyset$, now define a new NTHFA $M = \langle Q, \Sigma, \psi, q_0, \mathcal{F} \rangle$ where, 
    \begin{itemize}
		\item[(a)] $Q = S \cup P \cup \{q_0\}$ with $q_0 \notin (S \cup P)$.
		\item[(b)] For all $q \in Q$ we have,
		\begin{eqnarray}\label{eq:FinalUniao}
		\mathcal{F}(q) = \left\{\begin{array}{ll}	\mathcal{F}_1(q), & \hbox{if } q \in S\\ \mathcal{F}_2(q), & \hbox{if } q \in P\\	\mathcal{F}_1(s_0) \sqcup \mathcal{F}_2(p_0),  & \hbox{if } q = q_0\end{array}\right.\
		\end{eqnarray}
		\item[(c)] For all $q,q' \in Q$ and $a \in \Sigma$ we have,
		\begin{eqnarray}\label{eq:TransicaoUniao}
		\psi(q, a, q') = \left\{\begin{array}{ll}	\psi_1(q, a, q'), & \hbox{if } q, q' \in S\\ 	\psi_1(s_0, a, q'), & \hbox{if } q = q_0, q' \in S\\ \psi_2(q, a, q'), & \hbox{if } q,q' \in P\\	\psi_2(p_0, a, q'), & \hbox{if } q = p_0, q' \in P \\ \{0\}, & \mbox{otherwise}\end{array}\right.
		\end{eqnarray}
	\end{itemize}
	it is evident that $M$ is NTHFA. Now notice that,
	\begin{eqnarray}\label{eq:ComputacaoUniaoLambda}
		f_M(\lambda) & \stackrel{Rem.  \ref{rema:ValorParaLambda}}{=} & \mathcal{F}(q_0) \nonumber\\
		& \stackrel{Eq. (\ref{eq:FinalUniao})}{=} & \mathcal{F}_1(s_0) \sqcup \mathcal{F}_2(p_0)\\
		& \stackrel{Rem. \ref{rema:ValorParaLambda}}{=} & F_{M_1}(\lambda) \sqcup F_{M_1}(\lambda)  \nonumber\\
		& = & f_{M_1} \uplus f_{M_2} \nonumber
	\end{eqnarray}
	and for all $w \in \Sigma^+$ with $w = a_1a_2\cdots a_n$, by $(H1)$ we have that $\mathbb{H}$ is a monoid thus,
	\begin{eqnarray}\label{eq:ComputacaoUniao}
		\begin{array}{l}
			f_M(a_1a_2\cdots a_n) = \bigsqcup_{q_f \in Q}\Big(\widehat{\psi}(q_0, a_1a_2\cdots a_n, q_f) \otimes \mathcal{F}(q_f)\Big) = \\
			\bigsqcup \Big(\bigsqcup_{q_f \in Q-\{q_0\}}\Big(\widehat{\psi}(q_0, a_1a_2\cdots a_n, q_f) \otimes \mathcal{F}(q_f)\Big),  (\widehat{\psi}(q_0, a_1a_2\cdots a_n, q_0) \otimes \mathcal{F}(q_0))\Big).
		\end{array}
	\end{eqnarray}
	But by equations (\ref{eq:part2}) and (\ref{eq:TransicaoUniao}), and also by $(H2)$ it is clear that,
	\begin{eqnarray*}
		\widehat{\psi}(q_0, a_1a_2\cdots a_n, q_0) & = & \{0	\}.
	\end{eqnarray*}
	Therefore, 
	\begin{eqnarray}\label{eq:ComputacaoUniao2}
		\mathcal{L}_\mathcal{M}(a_1a_2\cdots a_n) & = & \bigsqcup \Big(\bigsqcup_{q_f \in Q-\{q_0\}}\Big(\widehat{\psi}(q_0, a_1a_2\cdots a_n, q_f) \otimes \mathcal{F}(q_f)\Big), \nonumber\\
		& & \ \ \ \ \ \ \ \ \ \ \  \big(\{0\} \otimes \mathcal{F}(q_0)\big)\Big) \nonumber \\
		& \stackrel{(H2)}{=} & \bigsqcup_{q_f \in Q-\{q_0\}}\Big(\widehat{\psi}(q_0, a_1a_2\cdots a_n, q_f) \otimes \mathcal{F}(q_f)\Big)\\
		& = & \bigsqcup \Big(\bigsqcup_{q_f \in S}\Big(\widehat{\psi}(q_0, a_1a_2\cdots a_n, q_f) \otimes \mathcal{F}(q_f)\Big), \nonumber\\
		& & \ \ \ \ \ \bigsqcup_{q_f \in P}\Big(\widehat{\psi}(q_0, a_1a_2\cdots a_n, q_f) \otimes \mathcal{F}(q_f)\Big)\Big) \nonumber\\
		& = & f_{M_1}(w) \sqcup f_{M_2}(w) \nonumber\\
		& = & f_1 \uplus f_2 (w) \nonumber.
	\end{eqnarray}
	Hence, by equations (\ref{eq:ComputacaoUniaoLambda}) and (\ref{eq:ComputacaoUniao2}),  $f_M = f_{M_1} \uplus f_{M_2}$, completing the proof.
\end{proof}

The result of the above theorem shows that the $\mathbb{H}$-union is a closure for the set $\mathbb{T}_\mathcal{R}$, and this result is generalized as follows.

\begin{col}\label{col:Uniao}
	Let $\{f_i\}_{i \in I}$ be a finite family of THFL such that $f_i \in \mathbb{T}_\mathcal{R}$ for all $i \in I$, then there exists a NTHFA $M$ such that $f_M = \biguplus_{i \in I} f_i$.
\end{col}

\begin{proof}
	Using induction on $I$ and the theorem \ref{teo:Uniao}.
\end{proof}

The next result presents a characterization of the languages computed by NTHFA, i.e., the next result establishes the sufficient and necessary conditions for a THFL $f$ belongs to $\mathbb{T}_\mathcal{R}$.

\begin{teo}\label{teo:Finitude}
	Let $f:\Sigma^* \rightarrow \mathbb{H}$ be a THFL. Then the following statements are equivalent.
	\begin{itemize}
		\item[$(i)$] $f \in \mathbb{T}_\mathcal{R}$.
		\item[$(ii)$] $R_f$ is finite and for each $k \in R_{f}$ the set $S^k_f$ is a regular language.
	\end{itemize}
\end{teo}

\begin{proof}
    $(i) \Rightarrow (ii)$ Suppose that $f:\Sigma^* \rightarrow \mathbb{H}$ belongs to $\mathbb{T}_\mathcal{R}$, thus there exists a NTHFA $M = \langle Q, \Sigma, \psi, q_0, \mathcal{F} \rangle$ such that $f = f_M$, i.e. for all $w \in \Sigma^*$ we have, 
    $$f(w) = \bigsqcup_{q_n \in Q} \Big(\widehat{\psi}(q_0, w, q_n) \otimes \mathcal{F}(q_n)\Big).$$
    Nevertheless, by definition of $\mathcal{M}$ we have that, $Q\times \Sigma \times Q$ is finite. Therefore, $R_{\psi}$ and $R_\mathcal{F}$ are finite sets. But by Remark \ref{rema:elementosGerados} we have that $R_{\widehat{\psi}}$ is finite. Hence, the set $R_f$ is finite. Now for each $k \in R_f$ we define an NFA $A_k = \langle Q, \Sigma, q_0, \delta_{k}, F_k \rangle$ where:
    \begin{eqnarray}\label{eq:Use1}
	    \delta_k(q, a) = p &\Longleftrightarrow& k \sqsubseteq  \psi(q, a, p)
	\end{eqnarray}
	and
	\begin{eqnarray}\label{eq:Use2}
	    q \in F_k  &\Longleftrightarrow& k \sqsubseteq  \mathcal{F}(q).
	\end{eqnarray}
	with $q, p \in Q$ and $a \in \Sigma$. Now we have that for any $w \in \Sigma^*$,
	\begin{eqnarray*}
		w \in L(A_k) & \Longleftrightarrow & \widehat{\delta_k}(q_0, w)  \in F_k\\
		& \Longleftrightarrow & \exists q_1, \cdots, q_{n-1}, q_n \in Q \mbox{, such that } q_n \in F_k \mbox{ and }\\
		& & \delta_k(q_0, a_1) = q_1, \cdots, \delta_k(q_{n-1}, a_n) = q_n\\
		& \stackrel{Eq. (\ref{eq:Use1}), (\ref{eq:Use2})}{\Longleftrightarrow }& \exists q_1, \cdots, q_{n-1}, q_n \in Q \mbox{, such that } k \sqsubseteq  \mathcal{F}(q_n)  \mbox{ and }\\
		& & k \sqsubseteq \psi(q_0, a_1, q_1), \cdots, k \sqsubseteq \psi(q_{n-1}, a_n, q_n)\\
		& \Longleftrightarrow & \exists q_n \in Q \mbox{, such that } k \sqsubseteq  \mathcal{F}(q_n) \mbox{ and } k \sqsubseteq \widehat{\psi}(q_0, w, q_n)\\
		& \stackrel{(R2)}{\Longleftrightarrow} & \exists q_n \in Q \mbox{, such that } k \sqsubseteq  \Big(\mathcal{F}(q_n) \otimes \widehat{\psi}(q_0, w, q_n)\Big)\\
		& \stackrel{Theo. \ref{teo:CotaSuperior}}{\Longleftrightarrow} & k \sqsubseteq  \bigsqcup_{q_n \in Q} \Big(\widehat{\psi}(q_0, w, q_n) \otimes \mathcal{F}(q_n)\Big)\\
		& \Longleftrightarrow & k \sqsubseteq f(w)\\
		& \Longleftrightarrow & w \in S^k_f
	\end{eqnarray*}
	Hence, $L(A_k) = S^k_\mathcal{L}$. Since $A_k$ is a NFA, by Theorem \ref{teo:regular2} we have that $ S^k_\mathcal{L}$ is a regular language.
	
	$(ii) \Rightarrow (i)$ Assume that $R_f = \{k_1, \cdots, k_n\}$ for some $n \in \mathbb{Z}_+^*$ and that for each $k \in R_f$ the set $S^{k}_\mathcal{L}$ is a regular language. Hence, by Theorem \ref{teo:regular} there exists a finite family of DFA $\{A_k\}_{k \in R_f}$, where $A_k = \langle Q_k, \Sigma, \delta_k, q^k_0, F_k \rangle$ such that $L(A_k) = S^k_f$. Now for each $k \in R_f$ define an NTHFA $M_k = \langle Q_k, \Sigma, \psi_k, q^k_0, \mathcal{F}_k \rangle$ such that for each $q,q' \in Q_k$ and $a \in \Sigma$:
	\begin{eqnarray}\label{eq:TransicaoK}
	    \psi_k(q, a, q') = \left\{\begin{array}{ll} \{1\}, & \mbox{if } \delta_k(q, a) = q'\\ \{0\}, & \mbox{else } \end{array}\right.
	\end{eqnarray}
	and 
	\begin{eqnarray}\label{eq:FinalK}
	    \mathcal{F}_k(q) = \left\{\begin{array}{ll} k, & \mbox{if } q \in F_k \\ \{0\}, & \mbox{else } \end{array}\right.
	\end{eqnarray}
	By the construction above we have that  for all $w \in \Sigma^*$,
	\begin{eqnarray*}
		w \in L(A_k) & \Rightarrow & f_{M_k}(w) = k
	\end{eqnarray*}
	and  
	\begin{eqnarray*}
		w \notin L(A_k) \Rightarrow f_{M_k}(w) = \{0\}
	\end{eqnarray*}
	so by definition $f_{\mathcal{M}_k} \in \mathbb{T}_\mathcal{R}$. Since $R_f$ is finite there exists a finite family $\{f_{M_k}\}_{k \in R_f}$ and, by corollary \ref{col:Uniao}, there exists an NHTFA $M$ such that,
	\begin{eqnarray*}
		f_M = \biguplus_{k \in R_f} f_{M_k}.
	\end{eqnarray*}
	Moreover, it is clear that $f_M = f$, therefore, $f \in \mathbb{T}_\mathcal{R}$.
\end{proof}

But by Theorem above, it is possible to conclude the following result.

\begin{col}\label{col:LimitacaoH-Regular}
	There exists THFL $f:\Sigma^* \rightarrow \mathbb{H}$ such that $f \notin \mathbb{T}_\mathcal{R}$.
\end{col}

\begin{proof}
	The THFL $f:\Sigma^* \rightarrow \mathbb{H}$, defined by:
	\begin{eqnarray*}
		f(w) =  \bigcup_{i = 0}^{|w|} \Big\{\frac{1}{(2^i + 1)}\Big\}
	\end{eqnarray*}
	is such that $R_f$ is infinite. Hence,  by Theorem  \ref{teo:Finitude} we have that $f \notin \mathbb{T}_\mathcal{R}$.
\end{proof}

The characterization present by Theorem  \ref{teo:Finitude} and the Corollary \ref{col:LimitacaoH-Regular} induce the inclusion described in figure \ref{Ima:hierarquiaH}.

\begin{figure}[h]
	\centering
	\begin{tikzpicture}
	\begin{scope}[shift={(3cm,-5cm)}, fill opacity=0.9]
	\draw[fill=aliceblue, draw = black] (0,0) circle (2.5);
	\draw[fill=anti-flashwhite, draw = black] (0,0) circle (1.3);
	\node at (0,1.8) (A) {$\mathbb{T}$};
	\node at (0,0.1) (D) {$\mathbb{T}_\mathcal{R}$};
	\end{scope}
	\end{tikzpicture}
	\caption{Inclusion between $\mathbb{T}$ and $\mathbb{T}_\mathcal{R}$.}
	\label{Ima:hierarquiaH}
\end{figure}
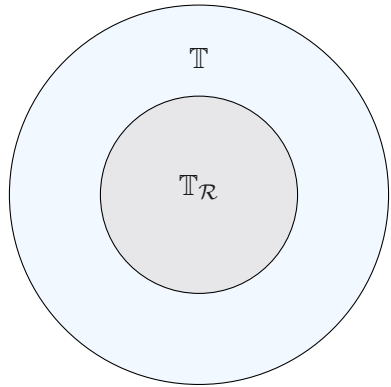

\begin{teo}\label{teo:CardinalidadeTHFL}
	The set $\mathbb{T}_\mathcal{R}$ is nondenumerable.
\end{teo}

\begin{proof}
    For each $x \in [0,1]$ define a THFA $\mathcal{M}_x = \langle \{q^x_0, q^x_1\}, \Sigma, \psi, q^x_0, \mathcal{F}_x \rangle$ such that for all $q,q' \in \{q^x_0, q^x_1\}$ and $a \in \Sigma$ we have: 
    \begin{equation*}
        \psi(q, a, q') = \{x\}
    \end{equation*}
    and
    \begin{equation*}
        \mathcal{F}_x(q) = \{x\}
    \end{equation*}
    now for each $w \in \Sigma^*$ it is clear that $f_{\mathcal{M}_x}(w) = \{x\}$. Now define the set $\Theta = \{f_{\mathcal{M}_x} \mid x \in [0,1] \}$,  clearly $\Theta \subset \mathbb{T}_\mathcal{R}$, moreover, there exists a bijection from $\Theta$ into $[0,1]$, so $\Theta$ is nondenumerable. Therefore, $\mathbb{T}_\mathcal{R}$ is nondenumerable.
\end{proof}

\section{Crisp Typical Hesitant Fuzzy Automata}

In this section, the paper is showing that the existence of a THFS of transitions in the definition of NTHFA is not essential to compute a THFL. For this, we will introduce below a new class of typical hesitant fuzzy automata.

\begin{defi}\label{defi:AFHTC}
	A Crisp Nondeterministic Typical Hesitant Fuzzy Automaton, or simply CNTHFA, is a quintuple $\mathcal{N} = \langle Q, \Sigma, \delta_\mathcal{N}, q_0, \mathcal{F} \rangle$, where $Q, \Sigma, q_o, \mathcal{F}$ is equal to definition \ref{defi:AFHT} and $\delta_\mathcal{N}: Q \times \Sigma \rightarrow 2^{Q}$ is equal to definition \ref{defi:AFN}.
\end{defi}

\begin{defi}
    A CNTHFA $M$ computes the THFL $f_M:\Sigma^* \rightarrow \mathbb{H}$ define by,
    \begin{equation}
        f_M(w) = \bigsqcup_{q \in \widehat{\delta_\mathcal{N}}(q_0, w)} \mathcal{F}(q)
    \end{equation}
\end{defi}

First, the research show that any THFL that an CNTHFA computes, also is computed by a NTHFA.

\begin{teo}\label{teo:LPseudoIsHregular}
	Let $\mathcal{N}$ be an CNTHFA, then there exists an NTHFA $\mathcal{M}$ such that $f_{\mathcal{N}} = f_\mathcal{M}$.
\end{teo}

\begin{proof}
	Given a CNTHFA $\mathcal{N} = \langle Q, \Sigma, \delta_\mathcal{N}, q_0, \mathcal{F} \rangle$, define a new NTHFA $\mathcal{M} = \langle Q, \Sigma, \psi, q_0, \mathcal{F} \rangle$ where:
	\begin{eqnarray}
	    \psi(q, a, q') = \left\{\begin{array}{ll} \{1\}, & \mbox{if } q' \in        \delta_\mathcal{N}(q,a)\\ \{0\}, & \mbox{else } \end{array}\right.
	\end{eqnarray}
	for all $q, q' \in Q$ and $a \in \Sigma$. Now For any $w \in \Sigma^*$ it is clear that $f_\mathcal{N}(w) = f_\mathcal{M}(w)$. Hence, $f_\mathcal{N} = f_\mathcal{M}$.
\end{proof}

On the other hand, the computational power of NTHFA is equivalent to the power of CNTHFA.

\begin{teo}\label{teo:HregularCautomaton}
	Let $\mathcal{M}$ be an NTHFA,  then there exist a CNTHFA $\mathcal{N}$ such that $\mathcal{N}$ has one more state than $\mathcal{M}$ and $f_\mathcal{M} = f_\mathcal{N}$.
\end{teo}

\begin{proof}
    Without loss of generality, by the proof of item  $(ii)$ in Theorem \ref{teo:Finitude} assume that $\mathcal{M} = \langle Q, \Sigma, \psi, q_0, \mathcal{F} \rangle$ is a NTHFA with the restriction $\psi(q,a,q') = \{1\}$ or $\psi(q,a,q') = \{0\}$ for all $q,q' \in Q$ and $a \in \Sigma$. Now define the following  CNTHFA  $\mathcal{N} = \langle Q \cup \{q_{\aleph}\}, \Sigma, \delta_\mathcal{N}, q_0, \mathcal{F}_\mathcal{N} \rangle$ where $q_{\aleph} \notin Q$, therefore, $\mathcal{N}$ has one more state than $\mathcal{M}$, moreover,  $\delta_\mathcal{N}$ is defined by the rules:
	\begin{itemize}
		\item[$(r1)$] If $\psi(q, a, q') = \{1\}$, then $q' \in \delta_\mathcal{N}(q, a)$ for any $q, q' \in Q$ and $a \in \Sigma$.
		\item[$(r2)$] If $\psi(q, a, q') = \{0\}$, then $q_{\aleph} \in \delta_\mathcal{N}(q, a)$ for any $q, q' \in Q$ and $a \in \Sigma$.
		\item[$(r3)$] $\delta_\mathcal{N}(q_{\aleph}, a) = \{q_{\aleph} \}$ for all $a \in \Sigma$.
	\end{itemize}
	Finally define $\mathcal{F}_\mathcal{N}$ as being:
	\begin{eqnarray}
	\mathcal{F}_\mathcal{N}(q) = \left\{\begin{array}{ll} \mathcal{F}(q), & \mbox{if } q \in Q\\ \{0\}, & \mbox{else } \end{array}\right.
	\end{eqnarray}
	By this construction, it is easy to see that $f_\mathcal{M}(w) = f_\mathcal{N}(w)$ for all $w \in \Sigma^*$, completing the proof.
\end{proof}

The above theorem shows that CNTHFA needs an extra state to compute the same language as an NTHFA. These results together present a new way to characterize the set $\mathbb{T}_\mathcal{R}$.

\begin{col}\label{col:C-LinguagemEregular}
	Let $f: \Sigma^* \rightarrow \mathbb{H}$ be a THFL. We have that $f \in \mathbb{T}_\mathcal{R}$ if and only if there exists CNTHFA $\mathcal{N}$ such that $f = f_\mathcal{N}$.	
\end{col}

\begin{proof}
	Straightforward  by theorems \ref{teo:LPseudoIsHregular} and \ref{teo:HregularCautomaton}.
\end{proof}

\begin{defi}\label{defi:AFHTDC}
	A Crisp Deterministic Typical Hesitant Fuzzy Automaton, or simply CDTHFA, is a quintuple $D = \langle Q, \Sigma, \delta, q_0, \mathcal{F} \rangle$, where $Q, \Sigma, q_o, \mathcal{F}$ is equal to definition \ref{defi:AFHT} and $\delta: Q \times \Sigma \rightarrow Q$ is equal to definition \ref{defi:AFD}. Let $D$ be a CDTHFA, the THFL computed by $D$ is exactly the THFL $f_D: \Sigma^* \rightarrow \mathbb{H}$ define as:
    \begin{equation}
        f_D(w) = \mathcal{F}(\widehat{\delta}(q_0, w))
    \end{equation}
\end{defi}

\begin{teo}\label{teo:Deter}
    For all CDTHFA $D$ there exists a CNTHFA $N$ such that $f_D = f_N$.
\end{teo}

\begin{proof}
    Is obvious, since CDTHFA can be seen as a special instance of CNTHFA with $\# \delta(q, a) \leq 1$ for all $(q, a) \in Q \times \Sigma$.
\end{proof}

\begin{teo}\label{teo:SemNaoDeterminismo}
    if $f \in \mathbb{T}_\mathcal{R}$, then there exists a CDTHFA $D$ such that $f = f_D$.
\end{teo}

\begin{proof}
    Suppose that $f \in \mathbb{T}_\mathcal{R}$, so there exists a CNTHFA $N = \langle Q, \Sigma, \delta_N, q_0, \mathcal{F}\rangle$ such that $f = f_N$, now it is sufficient to construct the CDTHFA $D = \langle 2^Q, \Sigma, \delta, \{q_0\}, \mathcal{F}_D \rangle$ where for all $(X, a) \in 2^Q \times \Sigma$ we have that:
    \begin{equation}
        \delta(X, a) = \bigcup_{q \in X}\delta_N(q, a)
    \end{equation}
    and
    \begin{eqnarray}
	    \mathcal{F}_D(X) = \left\{\begin{array}{ll} \displaystyle \bigsqcup_{q \in X} \mathcal{F}(q), & \mbox{if } X \neq \emptyset \\ \{0\}, & \mbox{else } \end{array}\right.
	\end{eqnarray}
	moreover, it is not difficult to verify that $\widehat{\delta}(\{q_0\}, w) = \widehat{\delta_N}(q_0, w)$ for all $w \in \Sigma^*$. Hence, we have that,
	\begin{eqnarray*}
	    f_D(w) & = & \mathcal{F}_D(\widehat{\delta}(\{q_0\}, w))\\
	    & = & \bigsqcup_{q \in \widehat{\delta}(\{q_0\}, w)} \mathcal{F}(q)\\
	    & = & \bigsqcup_{q \in \widehat{\delta_N}(q_0, w)} \mathcal{F}(q)\\
	    & = & f_N(w)\\
	    & = & f(w)
	\end{eqnarray*}
	completing the proof.
\end{proof}

The Theorem \ref{teo:SemNaoDeterminismo} shows that nondeterminism is not essential to compute THFL. Moreover, this lemma said that questions about $\mathbb{T}_\mathcal{R}$ elements could be seen as questions about CDTHFA.

\begin{defi}
    Let $f_1: \Sigma^* \rightarrow \mathbb{H}$ and $f_2: \Sigma^* \rightarrow \mathbb{H}$  two be THFL the $\mathbb{H}$-intersection of $f_1$ and $f_2$, denoted by $f_1 \Cap f_2$, is the THFL define by:
    \begin{eqnarray}
        f_1 \Cap f_2 (w) = f_1(w) \otimes f_2(w)
    \end{eqnarray}
\end{defi}

\begin{teo}\label{teo:inter}
    If $f_1, f_2 \in \mathbb{T}_\mathcal{R}$ on the same alphabet $\Sigma$, then $f_1 \Cap f_2 \in \mathbb{T}_\mathcal{R}$.
\end{teo}

\begin{proof}
    Suppose that $f_1, f_2 \in \mathbb{T}_\mathcal{R}$ so by Theorem \ref{teo:SemNaoDeterminismo} there exists two CDTHFAs $D_1 = \langle Q_1, \Sigma, \delta_1, q_0, \mathcal{F}_1\rangle$ and $D_2 = \langle Q_2, \Sigma, \delta_2, p_0, \mathcal{F}_2\rangle$ such that $f_1 = f_{D_1}$ and $f_2 = f_{D_2}$. Now without loss of generality assume that $Q_1 \cap Q_2 = \emptyset$, then define $D = \langle Q_1 \times Q_2, \Sigma, \delta, (q_0, p_0), \mathcal{F} \rangle$ where:
    \begin{eqnarray}
        \delta((q, p), a) = (\delta_1(q, a), \delta_2(p, a))
    \end{eqnarray}
    with $(q, p) \in Q_1 \times Q_2, a \in \Sigma$ and 
    \begin{equation}\label{eq:finalD}
        \mathcal{F}((q, p)) = \mathcal{F}_1(q) \otimes \mathcal{F}_2(p)
    \end{equation}
    it is easy to see that $D$ is a CDTHFA and that for all $(q, p) \in Q_1 \times Q_2$ and $w \in \Sigma^*$:
    \begin{equation}\label{eq:extendendoDelta}
        \widehat{\delta}((q, p), w) = (\widehat{\delta_1}(q, w), \widehat{\delta_2}(p, w))
    \end{equation}
    hence for all $w \in \Sigma^*$,
    \begin{eqnarray*}
        f_D(w) & = & \mathcal{F}(\widehat{\delta}((q_0, p_0), w)\\
        & \stackrel{Eq. (\ref{eq:extendendoDelta})}{=} & \mathcal{F}((\widehat{\delta_1}(q_0, w),  \widehat{\delta_2}(p_0, w)))\\
        & \stackrel{Eq. (\ref{eq:finalD})}{=} & \mathcal{F}_1(\widehat{\delta_1}(q_0, w)) \otimes \mathcal{F}_1(\widehat{\delta_2}(p_0, w))\\
        & = & f_1(w) \otimes f_2(w)\\
        & = & f_1 \Cap f_2 (w)
    \end{eqnarray*}
    since $D$ is a CDTHFA, by Theorems \ref{teo:Deter} and \ref{teo:SemNaoDeterminismo} $f_1 \Cap f_2 \in \mathbb{T}_\mathcal{R}$.
\end{proof}

The above theorem result shows that the $\mathbb{H}$-intersection is a closure for the set $\mathbb{T}_\mathcal{R}$ and this result can be generalized as follows.

\begin{col}\label{col:Intersecao}
	Let $\{f_i\}_{i \in I}$ be a finite family of THFL such that $f_i \in \mathbb{T}_\mathcal{R}$ for all $i \in I$, then there exists a CDTHFA $D$ such that $\displaystyle f_D = \bigdoublecap_{i \in I} f_i$.
\end{col}

\begin{proof}
    Using induction on $I$ and the theorem \ref{teo:inter}.
\end{proof}

\section{Conclusions}

This paper presents a characterization for the languages computed by nondeterministic typical hesitant fuzzy automata, i.e., here we prove the sufficient and necessary conditions for that a typical nondeterministic hesitant fuzzy automaton compute a typical hesitant fuzzy language. Besides, we show that typical hesitant fuzzy transitions and also that nondeterminism are not attributes necessary to compute typical hesitant fuzzy languages. This result corrects the previous result presented in \cite{valdigleis2019}, which presented the nondeterministic typical fuzzy automata as not equivalent to the deterministic counterpart.  This paper also presents the initial results about closure operators for the class of typical hesitant fuzzy languages computed by typical nondeterministic hesitant fuzzy automata. Here we prove that the $\mathbb{H}$-union and the $\mathbb{H}$-intersection are both closed for this class of languages. It is intended in future work to study the process of approximating typical hesitant fuzzy automata based on the idea of total admissible orders \cite{benja2021} and also based on partial order $\sqsubseteq$ introduced in this paper.

\nocite{*}
\bibliographystyle{fundam}
\bibliography{citations.bib}


\end{document}